\documentclass[a4paper,onecolumn,11pt]{quantumarticle}
\pdfoutput=1
\usepackage[utf8]{inputenc}
\usepackage[english]{babel}
\usepackage{amsmath}
\usepackage{mathtools}
\usepackage{amsfonts}
\usepackage{amssymb}
\usepackage{amsthm}
\usepackage{graphicx}
\usepackage{braket}
\usepackage{hyperref}
\usepackage{enumerate}

\DeclareMathOperator\AR{AR}
\DeclareMathOperator\tr{tr}
\DeclarePairedDelimiter\abs\lvert\rvert

\providecommand\onlinecite[1]{\cite{#1}}
\providecommand\openone{\mathbb{1}}

\providecommand\onlinecite[1]{\cite{#1}}
\providecommand\openone{\mathbb{1}}

\newtheorem{theorem}{Theorem}
\newtheorem{lemma}[theorem]{Lemma}
\newtheorem{corollary}[theorem]{Corollary}

\begin{document}

\title{Systematic construction of quantum contextuality scenarios with rank advantage}
\author{Pascal Höhn}

\author{Zhen-Peng Xu}
\email{zhen-peng.xu@uni-siegen.de}

\author{Matthias Kleinmann}
\email{matthias.kleinmann@uni-siegen.de}
\affiliation{Naturwissenschaftlich--Technische Fakultät, Universität Siegen, Walter-Flex-Straße 3, 57068 Siegen, Germany}

\begin{abstract}
A set of quantum measurements exhibits quantum contextuality when any consistent value assignment to the measurement outcomes leads to a contradiction with quantum theory. In the original Kochen--Specker-type of argument the measurement projectors are assumed to be rays, that is, of unit rank. Only recently a contextuality scenario has been identified where state-independent contextuality requires measurements with projectors of rank two. Using the disjunctive graph product, we provide a systematic method to construct contextuality scenarios which require non-unit rank. We construct explicit examples requiring ranks greater than rank one up to rank five.
\end{abstract}

\maketitle

\section{Introduction}

Quantum contextuality is a fundamental feature of quantum measurements \cite{Quantum_Contextuality} closely connected to the incompatibility of observables \cite{xu2019necessary} and quantum nonlocality \cite{cabello2010proposal, cabello2021converting}. It has been proved to be vital for quantum advantage like in quantum computation \cite{howard2014contextuality, raussendorf2013contextuality}, quantum communication \cite{cubitt2010improving, saha2019state, gupta2022quantum}, randomness generation \cite{miller2017universal}, quantum cryptography \cite{cabello2011hybrid} and quantum key distribution \cite{barrett2005no}. At the heart of quantum contextually are projective quantum measurements for which  one cannot construct a noncontextual hidden variable model. That is, if one assumes that a hidden variable determines the outcome of all measurements and that this value assignment is consistent across the different measurements, then the resulting predictions are at variance with the predictions of quantum theory. This definition follows the works by Kochen and Specker \cite{Kochen-Specker_117_rays} as well as Yu and Oh \cite{yu2012state} and differs from the notion of contextuality introduced by Spekkens \cite{spekkens2005contextuality}.

The original proof of quantum contextuality was found by Kochen and Specker \cite{kochen1975problem}. By using a set of 117 rank-one projectors in a three-dimensional space, they established a logical contradiction for quantum measurements under the hypothesis of a noncontextual hidden variable model. Similar proofs were gradually reported in different scenarios \cite{pavivcic2019hypergraph, cabello1996bell, kernaghan1995kochen, peres1991two}. As it turns out, a Kochen--Specker-like proof can always be converted into a state-independent noncontextuality inequality \cite{yu2015proof}, which is violated by every quantum state. However, the converse is not true. The first set of projectors that is not based on a Kochen--Specker-like proof but features state-independent contextuality (SIC) consists of 13 rank-one projectors in three dimensions \cite{yu2012state}. It is important to note that contrary to the related phenomenon of quantum nonlocality, the contradiction displayed by SIC is not a feature of the quantum state, but solely of the SIC set of measurement projectors. We focus here on the case of SIC, although there are also interesting instances where contextuality is state dependent \cite{klyachko2008simple}.

SIC sets involving only projectors of rank one have led to many fruitful results over the last years \cite{cabello2016quantum, Proof_Cabellos_18_rays_minimal, cabello2021converting}. In comparison, SIC sets involving non-unit rank received only little attention  \cite{mermin1990simple, kernaghan1995kochen, toh2013kochen, toh2013state} and, to the best of our knowledge, all of those results are based on the Mermin star \cite{Mermin_1993_Peres_Mermin_Square}.  But there is no physical reason to limit the projectors in a SIC set to be only of rank one and we expect that many interesting SIC sets involve projectors of non-unit rank. Indeed, it has been recently shown that \cite{State-ind-quantum-contextuality-with-projectors-of-nonunit-rank} non-unit rank can be indispensable for certain SIC scenarios, but the proof of this fact is specifically tailored to the specific SIC scenario and does not enable us to find other SIC scenarios with non-unit rank. In this paper we present systematic constructions of SIC scenarios that require projectors of non-unit rank. Our construction method uses known SIC sets and families of graphs with ``rank advantage.'' When such graphs are suitably combined by means of the disjunctive graph product, the resulting graph requires non-unit rank for SIC. To this end we provide a generic constriction method as well as special constructions based on numerical methods.

This paper is structured as follows. In Section~\ref{sec_sic_and_connection_to_graph_theory} we discuss the graph theoretical formulation of contextuality. We turn to SIC in Section~\ref{sec_state-independent_contextuality} where we also present a semidefinite program to compute the optimal SIC ratio. In Section~\ref{sec_about_AR_graphs} we introduce the rank advantage for graphs and present two families of graphs with this property. Subsequently we discuss the disjunctive graph product and its effect on graph-theoretical concepts in Section~\ref{sec_disjunctive_graph_product}. With this tools at hand we can formulate in Section~\ref{sec_construction_using_cycle_graphs} our main theorem and construct graphs for which we can exclude the existence of a SIC sets of low rank. In Section~\ref{sec_comb_of_SIC_graph_with_AR} we present further such constructions based on numerical methods. We conclude in Section~\ref{sec_conclusion}.

\section{Contextuality and its connection to graph theory}
\label{sec_sic_and_connection_to_graph_theory}
Our work uses the graph-theoretic framework of quantum contextuality \cite{Graph-Theoretic-Approach-to-Quantum-Correlations}. In this framework, for given projectors $(\Pi_i)_i$, the vertices of the corresponding exclusivity graph represent the projectors and two vertices are connected if the projectors are orthogonal, $\Pi_i\Pi_j=0$. Conversely, in a projective representation (PR) of a given graph one associates to each vertex a projector $\Pi_i$ such that  $\Pi_i\Pi_j=0$ holds for every edge $(i,j)$ of the graph. If the rank of all projectors in a PR is $r$, the representation is of rank $r$ and we denote by $d_\pi(G,r)$ the minimal dimension in which a rank-$r$ PR can be found for the graph $G$.

We now briefly review some important concepts. For a graph $G$ we write $V(G)$ for the set of vertices and $E(G)$ for the set of edges. A clique of a graph is any set of mutually connected vertices and an independent set is any set of vertices where no pair of vertices is connected. The cardinality of the maximal clique is the clique number $\omega(G)$ and clearly $d_\pi(G,r)\ge r\omega(G)$. The set of all independent sets is denoted by $\mathcal{I}(G)$ and the cardinality of the largest independent set is the independence number $\alpha(G)$. The chromatic number $\chi(G)$ is the minimal number of different colours that is needed to colour every vertex of a graph such that no vertices with the same colour are  connected. The fractional chromatic number $\chi_f(G)$ is a relaxation of the chromatic number. Here, there are in total $n$ colours and one assigns $m$ colours to every vertex. Then, the fractional chromatic number is the infimum over $\frac{n}{m}$, such that adjacent vertices do not share any of the assigned colours.

In a noncontextual hidden variable model, we assign outcomes $0$ or $1$ to each vertex of an exclusivity graph such that for adjacent vertices at most one outcome is $1$. The convex hull of all those possible  assignments is the set of noncontextual probability assignments. This coincides with the stable set \cite{Quantum_Contextuality, Graph-Theoretic-Approach-to-Quantum-Correlations}
\begin{align}
\label{eq_stable_set}
\text{STAB}(G) = \text{conv} \Set{ \boldsymbol v \in \set{ 0,1 }^{\abs{V(G)}} | v_i v_j = 0  \text{ for all }  (i,j) \in E(G) }.
\end{align}
The stable set gives rise to the weighted independence number\cite{grotschel1984polynomial}
\begin{equation}
\alpha(G,\boldsymbol{w})=\max_{\boldsymbol{x}} \set{ \boldsymbol{w} \cdot \boldsymbol{x} |\boldsymbol{x} \in \text{STAB}(G)} = \max_I \Set{ \sum_{i\in I}w_i | I\in \mathcal I(G) }
\end{equation}
for the vector of weights $\boldsymbol w\in \mathbb R^{\abs{V(G)}}$.

In quantum theory, outcomes of (sharp) measurements are described by projectors $\Pi_k$ and the probability for an event to occur is calculated as $p_k =\tr(\rho \Pi_k)$. The set of all possible quantum probability assignments is then given by the theta body \cite{State-ind-quantum-contextuality-with-projectors-of-nonunit-rank, Graph-Theoretic-Approach-to-Quantum-Correlations, GROTSCHEL_1986_Relaxation_of_vertex_packing}
\begin{align}
\label{eq_Theta_body}
    \text{TH}(G) = \set{  \left( \tr(\rho \Pi_i) \right)_i  | (\Pi_i)_i \text{ is a PR of $G$,}\; \tr(\rho) = 1,\; \rho \geq 0 }.
\end{align}
Quantum contextuality occurs now if one can find a probability assignment $P$ that can only be achieved by the quantum model, but not by any noncontextual hidden variable model, that is \cite{State-ind-quantum-contextuality-with-projectors-of-nonunit-rank}, $P \in \text{TH}(G) \setminus \text{STAB}(G)$.

\section{State-independent contextuality}
\label{sec_state-independent_contextuality}
In the previous section, the probability assignment $P$ depends on the choice of the quantum state and the PR of the exclusivity graph $G$. In contrast, for SIC one requires that one can find a PR $(\Pi_k)_k$ of  $G$ such that for no state $\rho$ the corresponding quantum probability assignment $P=(\tr(\rho\Pi_k))_k$ is in $\text{STAB}(G)$. Since the stable set is convex as well as the set of all quantum probability assignments with fixed projectors, one can then find weights $\boldsymbol w$ such that the inequality
\begin{align}
\label{eq_condition_for_SICx}
   \min_\rho\left( \sum_k w_k \tr(\rho \Pi_k)\right)\le \alpha(G,\boldsymbol{w}),
\end{align}
is violated \cite{State-ind-quantum-contextuality-with-projectors-of-nonunit-rank}. The SIC ratio $\eta$ for the projectors $(\Pi_k)_k$ is the maximal ratio of the left hand side and the right hand side of Eq.~\eqref{eq_condition_for_SICx}, where the ratio is maximised over all weights $\boldsymbol w$. Hence a PR of $G$ features SIC if and only if $\eta >1$. Further maximising $\eta$ over all rank-$r$ PRs of a graph $G$ yields the rank-$r$ SIC ratio $\eta(G,r)$ of $G$ and a graph has a rank-$r$ PR featuring SIC exactly when $\eta(G,r)> 1$.

The SIC ratio $\eta$ of given projectors $(\Pi_k)_k$ can be found by solving the semidefinite program
\begin{align}
\label{eq_Problem_to_find_optimal_weights}
\begin{split}
\text{max} \quad & \eta \\
\text{such that} \quad & \sum_k w_k \Pi_k \geq \eta \openone, \\
& \sum_{k\in I} w_k \leq 1 \quad \text{for all } I\in \mathcal I(G) \\
& w_k\ge 0 \quad\text{for all } k\in V(G).
\end{split}
\end{align}
We mention, that if we relax the program by taking the trace of both sides of the first condition, then the solution $\eta^*$ of the modified program satisfies $\eta^*d=r\chi_f(G)$, where $d=\tr(\openone)$ is the dimension of the Hilbert space and $r$ the rank of the projectors. Thus, we have
\begin{equation}\label{eq-sic-ratio-fcn}
r\chi_f(G)\ge \eta(G,r)d_\pi(G,r),
\end{equation}
tightening the previously established relations \cite{Necessary_sufficient_condition_for_SIC_meas_scenarios, Necessary-and-Sufficient-Condition-for-Q-State-Independent-Contextuality,
State-ind-quantum-contextuality-with-projectors-of-nonunit-rank} between $\chi_f$, SIC, and $d_\pi$.

In order to see that the solution $\eta$ and $\boldsymbol w$ of Eq.~\eqref{eq_Problem_to_find_optimal_weights} yields the SIC ratio, one only has to show that in Eq.~\eqref{eq_condition_for_SICx} negative weights cannot  increase the SIC ratio. First, the left hand side clearly does not become smaller. Second, the right hand side does not change. Indeed for any independent set $I$ and $\alpha_I= \sum_{k\in I} w_k$, we can remove all $k$ with $w_k<0$ from $I$, yielding the independent set $I'\subset I$. Then $\alpha_{I'}\ge \alpha_I$ and hence negative weights do not increase $\alpha(G,\boldsymbol w)=\max_I \alpha_I$.

In the following, we often start from a given graph and we aim to prove or disprove the existence of PR. This problem can be solved numerically using an alternating optimisation algorithm \cite{State-ind-quantum-contextuality-with-projectors-of-nonunit-rank} for finding a Gram matrix with rank $d$ and with all entries $0$ that correspond to adjacent vertices. If a Gram matrix can be found, then a rank-one PR in dimension $d$ can be constructed. We will refer to this method as the see-saw algorithm throughout. The approach can then be used to find rank-$r$ PRs by extending the initial graph $G$ to the graph $G^r$, for which every vertex is replaced by a clique of size $r$, see Ref.~\onlinecite{State-ind-quantum-contextuality-with-projectors-of-nonunit-rank} for details. This method is in particular useful when aiming for SIC, since the dimension of a PR featuring SIC is upper bounded by the fractional chromatic number via Eq.~\eqref{eq-sic-ratio-fcn} and $\eta(G,r)>1$.

\section{Graphs with rank advantage}
\label{sec_about_AR_graphs}
\begin{figure}
    \centering
    \includegraphics[width=0.7\linewidth]{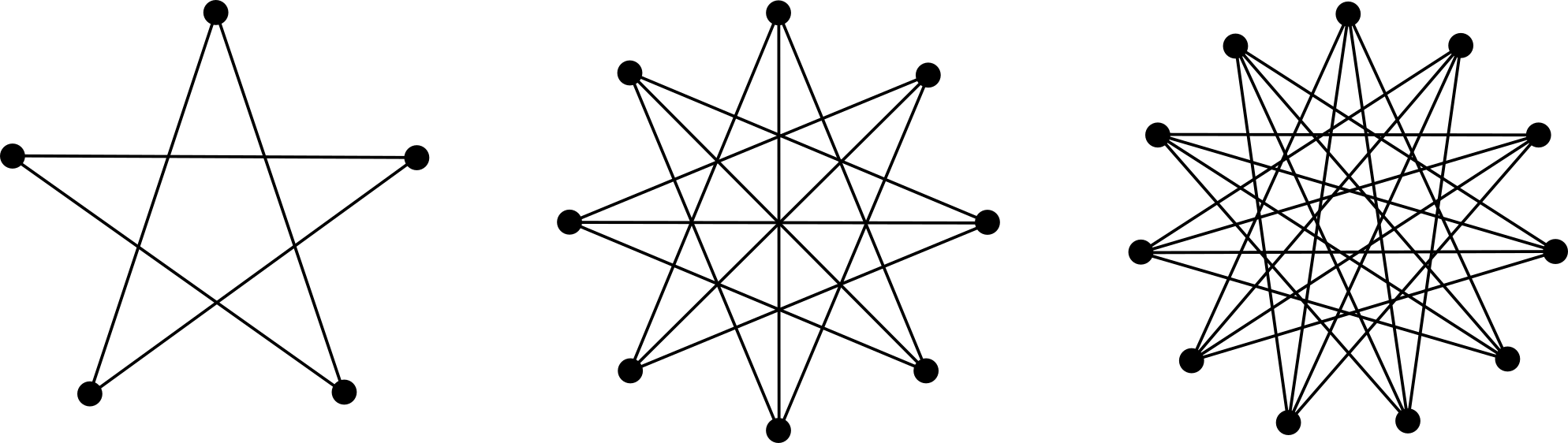}
    \caption{The first three examples of the $\AR(r)$ graphs, from left to right $\AR(2)$, $\AR(3)$ and $\AR(4)$.}
    \label{fig_AR_graphs}
\end{figure}

The fact that there exist SIC scenarios that require projectors of a higher rank \cite{State-ind-quantum-contextuality-with-projectors-of-nonunit-rank} motivates the examination of graphs that have a rank advantage, independent of whether those graphs admit a PR featuring SIC. A graph has a rank advantage if a higher rank PR is more ``efficient'' with respect to the Hilbert space dimension than any rank-one PR, that is, $d_\pi(G,r)<rd_\pi(G,1)$.

The first class of examples with a rank advantage are the odd cycle graphs $C_{2r+1}$ with $r\ge 2$. These graphs have vertices $0,1,\dotsc,2r$ and edges $(j, j\oplus 1)$, where ``$\oplus$'' denotes addition modulo $2r+1$. Furthermore they have the following properties.
\begin{enumerate}[(i)]\itemsep0pt
\item $\alpha(C_{2r+1})=r$ and $\chi_f(C_{2r+1})=2+\frac1r$, see Proposition~3.1.2 of Ref.~\onlinecite{Fractional_Graph_Theory}.
\item $d_\pi(C_{2r+1},1)= 3$, see Ref.~\onlinecite{Araujo_2013}.
\item The projectors
\begin{align}
\label{eq_cycle_graph_projectors_rank_k}
  \Pi_j =  \sum_{t=0}^{r-1} \ket{jr \oplus t}\!\bra{jr \oplus t},
\end{align}
form a rank-$r$ PR of $C_{2r+1}$ in dimension $2r+1$ with $\sum_j \Pi_j = r\openone$.
\end{enumerate}
The last statement can be verified by direct calculation.
From properties (ii) and (iii) it follows that $C_{2r+1}$ has a rank advantage for rank $r$.

The second class of examples uses circulant graphs, which are a generalisation of the cycle graphs. The circulant graph $C(n,(x_k)_k)$ has $n$ vertices $0,1,\dotsc,n-1$, such that each vertex $j$ is connected to the vertices $j\oplus x_k$ and $j\oplus (-x_k)$ for all $x_k$. Again, ``$\oplus$'' denotes addition modulo $n$. For $r\ge 2$ we define now the family
\begin{align}
\label{eq_AR_as_circulant_graphs}
\AR(r)=C(3r-1,(r,r+1,\dotsc,2r-1)),
\end{align}
see Figure~\ref{fig_AR_graphs}. The $\AR$ graphs have the following properties.
\begin{enumerate}[(i)]\itemsep0pt
\item $\alpha(\AR(r))=r$ and $\chi_f(\AR(r))=3-\frac1r$.
\item $d_\pi(\AR(r),1)= 3$.
\item The projectors
\begin{align}
\label{eq_projectors_of_AR_graphs}
    \Pi_j =  \sum_{t=0}^{r-1} \ket{j\oplus t}\!\bra{j\oplus t}.
\end{align}
form a rank-$r$ PR of $\AR(r)$ in dimension $3r-1$ such that $\sum_j \Pi_j = r\openone$.
\end{enumerate}
The proofs of properties (i) and (ii) are deferred to Appendix~\ref{app-AR-properties} and (iii) can be again verified by direct calculation. Combining (ii) and (iii) we see that $\AR(r)$ has a rank advantage for rank $r$. Using the see-saw algorithm for $r=2,\dotsc,10$, we observe that $\AR(r)$ does not have a rank advantage for any rank lower than $r$.

\section{The disjunctive graph product}
\label{sec_disjunctive_graph_product}
Our main tool to construct high-rank SIC scenarios is the disjunctive product $G\lor F$. It is defined via the Cartesian product of the vertex sets of two graphs $G$ and $F$,
\begin{align}
    V(G \lor F) = \set{ (u,v) \mid u \in V(G), v \in V(F) }
\end{align}
and the edge set given by
\begin{align}
    E(G \lor F) = \set{ ((u,v),({\tilde{u}},{\tilde{v}})) \mid (u,{\tilde{u}} ) \in E(G) \text{ or } (v,{\tilde{v}}) \in E(F)}.
\end{align}
The disjunctive product occurs naturally when considering the tensor product of PRs of two graphs. More specifically, for $(\Gamma_u)_u$ a PR of the graph $G$ and  $(\Gamma_v)_v$ for a PR of the graph $F$, the projectors $(\Gamma_u \otimes  \Omega_v)_{(u,v)}$ form a PR of $G\vee F$. This follows at once because $\Omega_u \otimes  \Omega_v$ is orthogonal to $\Gamma_{\tilde u} \otimes  \Omega_{\tilde v}$ exactly when $\Gamma_u$ and $\Gamma_{\tilde u}$ are orthogonal or $\Omega_{v}$ and $\Omega_{\tilde v}$ are orthogonal (or both).

As we aim to combine features of graphs, we consider now the weighted independence number $\alpha(G\lor F,\boldsymbol{w} )$ as well as the weighted fractional chromatic number \cite{Schrivjer_Combinatorial_Optimization_p1096/97, State-ind-quantum-contextuality-with-projectors-of-nonunit-rank} $\chi_f (G\lor F,\boldsymbol{w} )$. Those numbers are important for our purposes since $\alpha(G\lor F,\boldsymbol{w} )$ yields the classical bound for a contextuality scenario and $\chi_f (G\lor F)$ provides a bound on the maximal dimension in which a PR featuring SIC can exist.

\begin{theorem}
\label{thm_alpha_and_chi_disj_prod}
For two graphs $G$ and $F$, weights $\boldsymbol{w}^G \ge 0$ for $G$, $\boldsymbol{w}^F \ge 0$ for $F$, and $\boldsymbol w^{G\vee F}=(w_i^G w_j^F)_{ij}$ for $G\vee F$, it holds that
\begin{align}
    \alpha\left(G\lor F, \boldsymbol{w}^{G\vee F}\right) &= \alpha\left(G,\boldsymbol{w}^G\right)  \alpha\left(F,\boldsymbol{w}^F\right), \\
    \chi_f(G \lor F, \boldsymbol{w}^{G\vee F}) &= \chi_f(G, \boldsymbol{w}^G)  \chi_f(F, \boldsymbol{w}^F).
\end{align}
\end{theorem} \noindent
This theorem generalises Proposition~2.2 and Lemma~2.8 in Ref.~\onlinecite{Feige_randomized_graph_products}. The proof of Theorem~\ref{thm_alpha_and_chi_disj_prod} is given in Appendix~\ref{proof_independence_number_of_disj_prod}.

\begin{corollary}
\label{cor_higher_rank_SIC_rep_exists}
Let $( \Gamma_i )_i$ be projectors of rank $r_1$ with SIC ratio $\eta_1$, and $( \Omega_j )_j$ projectors of rank $r_2$ PR with SIC ratio $\eta_2$. Then $(\Gamma_i \otimes \Omega_j)_{i,j}$ are rank-$r_1 r_2$ projectors with SIC ratio $\eta_{12}\ge\eta_1\eta_2$.
\end{corollary}

\begin{proof}
We write $G$ for the exclusivity graph of $(\Gamma_i)_i$ and $\boldsymbol w^G$ for the optimal weights in Eq.~\eqref{eq_Problem_to_find_optimal_weights} achieving the SIC ratio $\eta_1$. Similarly, we define $F$, $\boldsymbol w^F$, and $\eta_2$ for $(\Omega_j)_j$. Then
\begin{equation}
\sum_{i} w_i^G\Gamma_i \otimes \sum_j w_j^F \Omega_j \ge (\eta_1\alpha(G,\boldsymbol w^G)\openone)\otimes (\eta_2\alpha(F,\boldsymbol w^F) \openone).
\end{equation}
Due to Theorem~\ref{thm_alpha_and_chi_disj_prod}, the right-hand side becomes $\eta_1\eta_2 \alpha(G\lor F,\boldsymbol w^{G\lor F})\openone$. The assertion follows, since $(\Gamma_i\otimes \Omega_j)_{i,j}$ forms a rank-$r_1r_2$ PR of $G\lor F$.
\end{proof}
We mention that  $\eta(G\lor F,r_1r_2)\ge \eta(G,r_1)\eta(F,r_2)$ is a direct consequence of Corollary~\ref{cor_higher_rank_SIC_rep_exists}.

\section{Systematic construction for SIC requiring higher rank}
\label{sec_construction_using_cycle_graphs}
We can make use of Corollary~\ref{cor_higher_rank_SIC_rep_exists} by combining a graph $G$ with SIC for rank one, $\eta(G,1)>1$, and a graph $F$ with rank advantage and $\eta(F,r)\ge 1$. The resulting graph $G\lor F$ will then have SIC in rank $r$, $\eta(G\lor F,r)>1$. However, in general, the combined graph $G\lor F$ may still allow SIC for a lower rank or even rank one, since the optimal PR and optimal weights are not necessarily of product form. We provide now a construction method for SIC graphs that provably need a set of projectors of a rank greater than a given rank $r$.

\begin{theorem}
\label{thm_cycle_combination_for_higher_rank}
For any graph $G$ with $0<\kappa<1$, where $\kappa=2(\chi_f(G)-\omega(G))$, and any rank $r<\frac 1\kappa$, choose a rank $k$ such that
\begin{equation}\label{eq-mthm-cond}
  k\ge \frac{r\chi_f(G)}{1-r\kappa}.
\end{equation}
Then $\eta(G\lor C_{2k+1},r)\le 1$, while $\eta(G\lor C_{2k+1},k)\ge \eta(G,1)$.
\end{theorem}

Note that $k>r$ due to $\kappa >0$ and $\omega(G)\ge 1$. If $\eta(G,1)>1$, then $G\lor C_{2k+1}$ has a rank-$k$ PR featuring SIC in dimension $d=(2k+1)\omega(G)$, simply by taking the tensor product of the rank-one PR of $G$ and the rank-$k$ PR of $C_{2k+1}$ from Eq.~\eqref{eq_cycle_graph_projectors_rank_k}. The examples constructed from Theorem~\ref{thm_cycle_combination_for_higher_rank} necessarily have a SIC ratio close to unity due to $\eta(G,1)\le \frac{\chi_f(G)}{\omega(G)}= 1+\frac{\kappa}{2\omega(G)}<1+\frac1{2r\omega(G)}$. For the proof of Theorem~\ref{thm_cycle_combination_for_higher_rank} we utilise the following lower bound on the dimensions of a PR of $G\lor C_{2k+1}$.

\begin{lemma}
\label{lem_disj.prod._of_SIC_and_odd_cyc_has_no_rep_in_dim_2_times_d}
For any graph $G$ and odd $\ell>1$ we have $d_\pi(G\lor C_\ell,r)\ge 2r\omega(G)+1$.
\end{lemma}

The proof of Lemma~\ref{lem_disj.prod._of_SIC_and_odd_cyc_has_no_rep_in_dim_2_times_d} can be found in Appendix~\ref{app_odd_cycle_properties}.

\begin{proof}[Proof of Theorem~\ref{thm_cycle_combination_for_higher_rank}]
According to Theorem~\ref{thm_alpha_and_chi_disj_prod} and using $\chi_f( C_{2k+1}) = 2 +\frac{1}{k}$ we have
\begin{align}
    \chi_f(G \lor C_{2k+1}) = \chi_f(G) \chi_f( C_{2k+1}) = \left(2 +\frac{1}{k}\right) \chi_f(G).
\end{align}
Now, for any $r<\frac1\kappa$ and $k$ obeying Eq.~\eqref{eq-mthm-cond}, we get
\begin{equation}\begin{split}
    d_\pi(G\lor C_{2k+1},r)\eta(G\lor C_{2k+1},r)&\le r \chi_f(G \lor C_{2k+1})\\& \le 2r\omega(G)+1 \le d_{\pi}(G \lor C_{2k+1},r),
\end{split}\end{equation}
where the first inequality corresponds to Eq.~\eqref{eq-sic-ratio-fcn} and the last inequality is due to Lemma~\ref{lem_disj.prod._of_SIC_and_odd_cyc_has_no_rep_in_dim_2_times_d}. Hence $\eta(G\lor C_{2k+1},r)\le1$.

For the second statement we use the existence of a rank-$k$ PR of $C_{2k+1}$ with $\sum \Pi_j=k\openone$ and $\alpha(C_{2k+1})=k$. Hence $\eta(C_{2k+1},k)\ge 1$. According to Corollary~\ref{cor_higher_rank_SIC_rep_exists} this implies $\eta(G\lor C_{2k+1},k)\ge \eta(G,1)$.
\end{proof}

Due to Theorem~\ref{thm_cycle_combination_for_higher_rank}, we can construct graphs which require a higher rank than a given $r$ for SIC. For this we need to find a graph with SIC and with $\kappa<\frac1r$. In Table~\ref{tab_thm_examples} we list examples of such graphs for $r=1,\dotsc,5$. Interestingly, most of the well-known SIC scenarios are not suitable for application of Theorem~\ref{thm_cycle_combination_for_higher_rank} or are only suitable for small $r$. Therefore we construct three sets of rank-one SIC sets that are particularly well-suitable for Theorem~\ref{thm_cycle_combination_for_higher_rank}. The first two sets, see Table~\ref{tab_Pascal_18_rays}, have 18 rays in dimension 4 and thereby are similar to the 18 rays by Cabello \cite{cabello1996bell}. For the third set, we use 21 rays found by Bengtsson, Blanchfield, and Cabello \cite{BBC_21_complex_vectors} and remove rays from the set to find the smallest set with SIC ratio $\eta>1$. At most 4 rays can be removed, for example the rays generated by $(0,1,-q)$, $(-1,0,1)$, $(1,q,q)$, and $(q,1,q)$, where $q=\mathrm e^{2\pi i/3}$. The resulting graph $G_\mathrm{BBCr}$ performs best with respect to Theorem~\ref{thm_cycle_combination_for_higher_rank} by allowing up to $r=5$.

\begin{table}
    \centering
    \setlength{\tabcolsep}{8pt}
    \renewcommand{\arraystretch}{1.2}
    \begin{tabular}{lcccccrrrrr}
         \hline
        Graph $G$ & $\abs{V(G)}$ &  $\chi_f(G)$ & $\omega (G)$ & $\eta(G,1)$ & $1 / \kappa$ & $k_1$ & $k_2$ & $k_3$ & $k_4$ & $k_5$ \\
        \hline
        $G_{\mathrm{BBC}}$ & 21 & $3+\frac{1}{3}$  & 3 & $1+\frac{1}{9}$  &  $1+\frac{1}{2}$ & 10 & --  &  -- &  -- &  -- \\
        $G_{\mathrm{YO}}$ &  13 & $3+\frac{2}{11}$ & 3 & $1+\frac{2}{33}$ & $2+\frac{3}{4}$ &  5 & 24  & --  &  -- & -- \\
        $G_{\mathrm{H}}$  & 18 & $4+\frac{1}{8}$ & 4 & $1+\frac{1}{75}$ & 4 & 6  & 17  & 50  & --  & -- \\
        $G_{\mathrm{X}}$ & 18 & $4+\frac{1}{7}$ & 4 & $1+\frac{1}{42}$ & $3 + \frac{1}{2}$ &  6 & 20  & 87  & --  & --  \\
        $G_{\mathrm{BBCr}}$ & 17 & $3+\frac{1}{11}$ & 3 & $1+\frac{1}{78}$ & $5+\frac{1}{2}$ & 4  & 10  & 21  & 46  & 170  \\
         \hline
    \end{tabular}
    \caption{List of graphs $G$ and values for $k_r$, for which $G\lor C_{2k_r+1}$  does not have a rank-$r$ PR featuring SIC while it does so for rank $k_r$. The ranks $r$ and $k_r$ are chosen according to Theorem \ref{thm_cycle_combination_for_higher_rank}. Here, $G_\mathrm{BBC}$ denotes the graph formed by the 21 rays found by Bengtsson, Blanchfield, and Cabello in Ref.~\onlinecite{BBC_21_complex_vectors} and $G_\mathrm{YO}$ the graph by Yu and Oh formed by 13 rays, see Ref.~\onlinecite{yu2012state}. The graphs $G_{\mathrm{H}}$ and $G_{\mathrm{X}}$ correspond to the exclusivity graphs of the rays given in Table~\ref{tab_Pascal_18_rays} (a) and (b), respectively. Finally, $G_{\mathrm{BBCr}}$ is obtained by removing 4 vertices from $G_{\mathrm{BBC}}$. The graph6-codes for the graphs are provided in Appendix~\ref{appendix-graph6}. For the graphs $G_\mathrm{H}$, $G_\mathrm{X}$, and $G_\mathrm{BBCr}$, the value given for $\eta(G,1)$ is a lower bound.}
    \label{tab_thm_examples}
\end{table}

\begin{table}
    \centering (a)
    \setlength{\tabcolsep}{8pt}
    \begin{tabular}{|c|c|c|c|c|c|c|c|c|c|c|c|c|c|c|c|c|c|}
         \hline
         \rule{0pt}{7pt}1  & 0  &  0   &  1  &  1  &  1  &  1  &  0  &  0  & 0  &  0   &  1  &  1 &  1  &  0  &  0  &  0  &  0    \\
         0 & 0 & 0  & 1 &$\overline{1}$ & 0 & 0 & 1 & 1 & 0 & 0  & 1 &$\overline{1}$ &$\overline{1}$ & 1 & 1 & 1 & 1   \\
         0 & 1 & 0  & 0 & 0 & 1 &$\overline{1}$ & 0 & 0 & 1 & 1  &$\overline{1}$ & 1 &$\overline{1}$ & 1 & 1 &$\overline{1}$ &$\overline{1}$    \\
         0 & 0 & 1  & 0 & 0 & 0 & 0 & 1 &$\overline{1}$ & 1 &$\overline{1}$  & 0 & 0 & 0 & 1 &$\overline{1}$ & 1 &$\overline{1}$   \\
         \hline
    \end{tabular}\\ \medskip (b)
    \begin{tabular}{|c|c|c|c|c|c|c|c|c|c|c|c|c|c|c|c|c|c|}
         \hline
        \rule{0pt}{7pt}0 & 0 &  0 &  1  &  1  &  1  &  1  & 1  &  0  &  0  &  0  &  0  & 1 & 1 & 1 & 1 & 1 & 1   \\
         1 & 0 & 0  & 1 & 0 & 0 & 0 & 0 & 1 & 1 & 1 & 1  &$\overline{1}$ &$\overline{1}$ &$\overline{1}$ &$\overline{1}$ & 1 & 1    \\
         0 & 1 & 0  & 0 & 0 & 0 & 1 &$\overline{1}$ & 0 & 0 & 1 &$\overline{1}$   & 0 & 0 & 1 &$\overline{1}$ & 1 & 1    \\
         0 & 0 & 1  & 0 & 1 &$\overline{1}$ & 0 & 0 & 1 &$\overline{1}$ & 0 & 0   & 1 &$\overline{1}$ & 0 & 0 & 1 &$\overline{1}$    \\
         \hline
    \end{tabular}
    \caption{Two sets of 18 rays that form a SIC set in dimension 4. The exclusivity graph of (a) is $G_{\textrm H}$, while the exclusivity graph of (b) is $G_{\textrm{X}}$. We abbreviate $-1$ with  $\overline{1}$.}
    \label{tab_Pascal_18_rays}
\end{table}

\section{Combination of SIC graphs with $\AR(r)$ graphs}
\label{sec_comb_of_SIC_graph_with_AR}
In contrast to the general construction in Theorem~\ref{thm_cycle_combination_for_higher_rank} we aim now to find specific examples with special properties. In particular we aim for rank efficient graphs in the sense that rank $r$ is sufficient for SIC, while for $r-1$ no representation with SIC can be found. We achieve this by considering the disjunctive graph product with the advantage rank graphs $\AR(r)$. We have already seen in Section~\ref{sec_about_AR_graphs} that $d_{\pi}(\AR(r),1) = 3$ and that there exists a rank-$r$ representation in dimension $3r-1$. In contrast to the cycle graphs, an $\AR(r)$ graph has a rank advantage only for rank $r$ or larger, as we have numerically verified using the see-saw algorithm for $r\le 10$. Therefore, we combine those graphs with small SIC graphs and our numerical results suggest that the combined graphs have SIC for rank $r$ but they do not have SIC for lower rank.

\begin{table}
    \centering
    \setlength{\tabcolsep}{8pt}
    \renewcommand{\arraystretch}{1.2}
    \begin{tabular}{lcrccr}
        \hline
        Graph $\mathcal G$ & $r$ & $\abs{V(\mathcal G)}$ & $\chi_f(\mathcal G)$ & $\eta(\mathcal G,r)$ & $d$ \\
        \hline
        $G_\mathrm{YO} \lor \AR(2)$ & 2 &  65 & $7+\frac{21}{22}$ & $1+\frac{2}{33}$ & 15 \\
        $\mathcal R(G_\mathrm{YO} \lor AR(2))$ & 2 &  39 & $7 + \frac{71}{88}$ & $1+\frac{2}{65}$ & 15 \\
        $G_\mathrm{CEG}\lor \AR(2)$ & 2 &  90 & $11+\frac{1}{4}$ & $1+\frac{1}{8}$ & 20 \\
        $\mathcal R(G_\mathrm{CEG}\lor AR(2))$  & 2 &  54 & $10 + \frac{3}{4}$ & $1+\frac{1}{17}$ & 20 \\
        $G_\mathrm{YO} \lor \AR(3)$ & 3 & 104 & $8+\frac{16}{33}$ & $1+\frac{2}{33}$ & 24 \\
        \hline
    \end{tabular}
    \caption{Examples of graphs which have a rank-$r$ PR featuring SIC, while we find that this is not the case for any lower rank. The evidence for the non-existence of a PR is obtained through the see-saw algorithm. Here, $G_\mathrm{YO}$ is the graph by Yu and Oh \cite{yu2012state} and $G_\mathrm{CEG}$ denotes the graph due Cabello, Estebaranz, and Garc\`{i}a-Alcaine  \cite{cabello1996bell}. The reduced graphs, $\mathcal R(G)$, are obtained by systematically removing elements from the PR while maintaining a SIC ratio above unity. The value for the SIC ratio $\eta(G,r)$ is a lower bound, obtained for a PR in dimension $d$.}
    \label{tab_ar_examples}
\end{table}

Our results are summarised in Table~\ref{tab_ar_examples}. For every of the examples, we determine a lower bound on the SIC ratio by using the PR obtained from the Kronecker product of the rank-one PR of the graph and the PR of $\AR(r)$ in Eq.~\eqref{eq_projectors_of_AR_graphs}. We then use the see-saw algorithm to exclude a PR with rank $r-1$ in dimension $d'< (r-1)\chi_f(G)$. A failure of the algorithm with $40.000$ random initial values gives us a strong indication that no such PR exists. Additionally, we reduce the graphs $G\lor \AR(r)$ using the following heuristic method. For a given PR with SIC we check if any projector can be removed while maintaining a SIC ratio above unity, $\eta>1$. If this is the case, we remove the projector which still yields the highest SIC ratio. Iterating this, we obtain the reduced exclusivity graph $\mathcal R(G\lor \AR(r))$. We then use again the see-saw algorithm to exclude a PR of $\mathcal R(G\lor\AR(r))$ of lower rank.

\section{Conclusion}
\label{sec_conclusion}
Quantum measurements have a multitude of ways to be nonclassical with SIC being a prominent case. The key structure here are the exclusivity relations between the measurement outcomes, that is, which measurement projectors are orthogonal. We showed that the rank required for a SIC PR of given exclusivity relations can be arbitrarily large, as long as one can find a corresponding graph $G$ which satisfies the conditions of Theorem~\ref{thm_cycle_combination_for_higher_rank}. Based on this we used a graph which allows us to construct exclusivity relations that do not have a SIC PR if the rank is $r$ or lower, for $r=1,2,3,4,5$. We conjecture that one can find similar examples for any rank $r$. Compared to the only example known previously, the construction is simple and generic. We mention that the resulting graphs require a high dimension of the Hilbert space, for example, $d=1023$ for rank 5. We constructed significantly smaller examples using the disjunctive product and numerical methods. As quantifier for the strength of the SIC we use the SIC ratio, which can be computed for a given PR by solving the semidefinite optimisation problem in Eq.~\eqref{eq_Problem_to_find_optimal_weights}. In particular, the SIC ratio of our scenario excluding rank 2 has SIC ratio $1+\frac19$ and outperforms the SIC ratio of $1+\frac1{14}$ for the scenario by Toh \cite{toh2013state}.

In summary, we have shown that higher rank projectors are essential to SIC and thus to our understanding of the structure of quantum measurements. Our methods are open to extensions to the inhomogenous case and according generalisations will be subject to future research. More generally, it remains an open problem to  find advantages of SIC with high rank over SIC with rank one, for example, with respect to number of projectors in the SIC set or with respect to the violation.

\cleardoublepage
\appendix

\section{Properties of the $\AR$ graphs}\label{app-AR-properties}
Here we show that the graphs $\AR(r)$ defined in Eq.~\eqref{eq_AR_as_circulant_graphs} of the main text  have $d_\pi(\AR(r),1)= 3$, $\alpha(\AR(r))=r$, and $\chi_f(\AR(r))=3-\frac1r$.

We first show that any rank-one PR of $\AR(r)$ needs at least dimension $3$ by showing that this is already true for a subgraph $S$ of $\AR(r)$. By the definition of $\AR(r)$, every node $v$ is connected to node $v\oplus r$ and node $v\oplus(r+1)$.
Now,
\begin{align}
    \set{0,r,2r,3r,4r} = \set{0,r,2r,1,r+1} \mod{3r-1}
\end{align}
and thus, we found the cycle graph with five nodes $C_5$ as subgraph of $\AR(r)$. It is well-known \cite{klyachko2008simple} that $d_\pi(C_5,1)=3$. Conversely, consider the (highly degenerate) assignment of projectors $\Pi_v=\ket{\lfloor v/r\rfloor}\!\bra{\lfloor v/r\rfloor}$, where $\lfloor x\rfloor$ denotes the largest integer not greater than $x$. Every projector $\Pi_v$ is orthogonal to the projectors $\Pi_{v\oplus r},\dots,\Pi_{v\oplus (2r-1)}$ and $0\le \lfloor v/r \rfloor\le 2$. Thus, $d_\pi(\AR(r))\le 3$.

Next we determine the independence number of $\AR(r)$. One easily verifies that $\set{0,\dotsc,r-1}$ is an independent set and we now show that no independent set can contain more than $r$ vertices. First, due to the symmetry of the problem, without loss of generality, we can require that the largest independent set $I$ contains the vertex $v=r-1$. Secondly, $I$ must be a subset of $\set{0,1,\dotsc, 2r-2}$ since all other vertices are already connected to the vertex $v=r-1$. Let $v_0=\min I$ and $v_1=\max I$. Then $v_1-v_0< r$, because $v_0$ is connected to all vertices $v_0+r,\dotsc,v_0+2r-1$. Hence, $\abs{I}\le r$ and by virtue of the above example, $\alpha(\AR(r))=r$. The value of the fractional chromatic number follows at once by using that $\chi_f(G)=\abs{V(G)}/\alpha(G)$ holds for vertex transitive graphs, see Proposition~3.1.1 in Ref.~\onlinecite{Fractional_Graph_Theory}.

\section{Proof of Theorem~\ref{thm_alpha_and_chi_disj_prod}}
\label{proof_independence_number_of_disj_prod}
In the following we give the proof of Theorem~\ref{thm_alpha_and_chi_disj_prod}.
In order to do so, we remember that the edge set of the disjunctive product $G \lor V$ is given by
\begin{align}
    E(G \lor F) = \set{ ((u,v),({\tilde{u}},{\tilde{v}})) |  \text{  where  } (u,{\tilde{u}} ) \in E(G) \ \lor \ (v,{\tilde{v}}) \in E(F) }.
\end{align}
It follows that two vertices are not connected, if and only if both vertices of the initial graphs are not connected, that is,
\begin{align}
    \label{eq_disj.prod.condition.edgesnotconnected}
      ((u,v),(\tilde{u},\tilde{v})) \notin E(G \lor F) \Leftrightarrow (u,\tilde{u} ) \notin E(G) \ \land \ (v,\tilde{v}) \notin E(F).
\end{align}
In order to proof Theorem~\ref{thm_alpha_and_chi_disj_prod}, we will need the following Lemma.
\begin{lemma}
\label{lemma_max_ind_set_of_G_or_F}
Every maximal independent set of $G \lor F$ is a direct product of elements of a maximal independent set of $G$ and a maximal independent set of $F$.
\end{lemma}
\begin{proof}
First, consider an arbitrary independent set $I$ of $G \lor F$ with elements of the form $(u,v)\in I$. We can order the elements of any arbitrary set in a matrix. This matrix will have as many rows as there are different $u$ in $I$ and as many columns as there are different $v$ in $I$. Let there be $n$ different $u$ and $m$ different $v$, then the matrix is similar to the form
\begin{align}
 \left(
\begin{array}{cccc}
(u_0,v_0)& (u_0,v_1) &  \cdots & (u_0,v_{m-1}) \\
(u_1,v_0)& (u_1,v_1) &  \cdots & (u_1,v_{m-1}) \\
\vdots & \vdots & \ddots & \vdots \\
(u_{n-1},v_0)& (u_{n-1},v_1) &  \cdots & (u_{n-1},v_{m-1}) \\
\end{array}
\right).
\end{align}
There can be holes in the matrix if $(u_i,v_j)\notin I$. However then there is some other element in row $i$ and in column $j$. It follows that $(u_i,v_j)$ can be added to $I$ as can be seen by the following explanation.

Without loss of generality we consider the case $(u_0,v_1) \in I$ and $(u_1,v_0) \in  I$, but $(u_0,v_0) \notin I$. From Eq.~\eqref{eq_disj.prod.condition.edgesnotconnected} we obtain
\begin{align}
     ((u_0,v_1),(u_1,v_0)) \notin E(G \lor F) \Leftrightarrow (u_0,u_1) \notin E(G) \ \land \ (v_0,v_1) \notin E(F).
\end{align}
This readily implies that none of $((u_0,v_0),(u_0,v_1))$, $((u_0,v_0),(u_1,v_0))$, or $((u_0,v_0),(u_1,v_1))$ can be an edge of $G \lor F$. Hence we can  add $(u_0,v_0)$ to $I$ and the resulting set is still an independent set.

Consequently, any maximal independent set $I$ is a set of the form $\{(u_i,v_j)\}_{i,j}$ and hence must be a direct product of independent sets of $G$ and $F$. Clearly, those independent sets of $G$ and $F$ must all be chosen to be maximal.
\end{proof}

Now we can proceed to proof Theorem~\ref{thm_alpha_and_chi_disj_prod}. We start with the proof for the weighted independence number. Let $\{\hat{u}\}_{\hat{u}}$ be an independent set of $G$ that fulfils $\sum_{\hat{u}} w^G_{\hat{u}} = \alpha(G,\boldsymbol{w}^G) $ and let $\{\hat{v}\}_{\hat{v}}$ be an independent set of $F$ that fulfils $\sum_{\hat{v}} w^F_{\hat{v}} = \alpha(G,\boldsymbol{w}^F)$.
By definition of the disjunctive graph product, the set $\{(\hat{u},\hat{v})\}_{\hat{u} \hat{v}}$ is an independent set. Since the weights $\{w_{\hat{u} \hat{v}}\}$ are products of the weights of $V(G)$ and $V(F)$, we have that the sum of the weights of the above independent set is
\begin{align}
    \sum_{\hat{u},\hat{v}} w_{\hat{u}\hat{v}} = \sum_{\hat{u},\hat{v}} w^G_{\hat{u}}  w^F_{\hat{v}} = \sum_{\hat{u}} w^G_{\hat{u}}  \sum_{\hat{v}} w^F_{\hat{v}} = \alpha(G,\boldsymbol{w}^G)  \alpha(F,\boldsymbol{w}^F).
\end{align}
Therefore, we found a lower bound on the weighted independence number of $G \lor F$
\begin{align}
\label{eq_disj.prod.condition.lowerboundonalpha(GorFw)}
    \alpha(G,\boldsymbol{w}^G)  \alpha(F,\boldsymbol{w}^F) \leq \alpha(G \lor F, \boldsymbol{w}).
\end{align}

Using Lemma~\ref{lemma_max_ind_set_of_G_or_F}, we can estimate the sum of the weights of any independent set of $G\lor F$.
One obtains that for every independent set $\{(u,v)\}_{u,v(u)}$ the sum of its weights is upper bounded by $\alpha(G,\boldsymbol{w}^G)  \alpha(F,\boldsymbol{w}^F)$, since
\begin{align}
\label{eq_app_weight_abschätzung}
    \sum_{u,u(v)} w_{u,v} \leq \sum_{u \in I,v \in J} w_{u,v} = \sum_{u \in I,v \in J} w^G_u  w^F_v = \sum_{u \in I} w^G_u  \sum_{v \in J} w^F_v \leq \alpha(G,\boldsymbol{w}^G)  \alpha(F,\boldsymbol{w}^F),
\end{align}
where $I$ is the independent set given by all different $u$ and $J$ is the independent set given by all different $v$.
 Since we know from Eq.~\eqref{eq_disj.prod.condition.lowerboundonalpha(GorFw)} that the independence number of $G \lor F$ is also lower bounded by $\alpha(G,\boldsymbol{w}^G)  \alpha(F,\boldsymbol{w}^F)$, we conclude that
\begin{align}
\label{eq_weighted_independence_number_of_disj._graphproduct}
     \alpha(G\lor F, \boldsymbol{w}) = \alpha(G,\boldsymbol{w}^G)  \alpha(F,\boldsymbol{w}^F).
\end{align}

Next, we consider the proof for the weighted fractional chromatic number. First, the weighted fractional chromatic number $\chi_f(G \lor F, \boldsymbol{w})$ is the solution to the problem \cite{Schrivjer_Combinatorial_Optimization_p1096/97}
\begin{align}
\label{eq_chiGorFw}
\begin{split}
\min\quad & \sum_{I} z_I \\
\text{such that} \ \ \ &\sum_{ I\ni i} z_I \geq w_i \quad \forall i \in V(G \lor F) \\
&  z_I  \geq 0,
\end{split}
\end{align}
where $I$ exhausts all of $\mathcal{I}(G \lor F)$. It is sufficient to only consider maximal independent sets as the optimal solution remains the same. This follows from the observation that in the above linear program we assign positive values $z_I$ to the independent sets, such that the sum of all sets in which a vertex $i$ is included, is greater equal the weight of this vertex. Thus, if an optimal solution is achieved by a solution that includes a non-maximal independent set $I$, the same solution is obtained by assigning $z_I$ to a maximal independent set that includes $I$ and then setting $z_I$ to zero. In the following, $\mathcal{I}_\mathrm{max}(G)$ will denote the set of maximal independent sets of $G$.

According to Lemma~\ref{lemma_max_ind_set_of_G_or_F}, every maximal independent set of $G \lor F$ is of form $I\times J$ for $I\in \mathcal I_\mathrm{max}(G)$ and $J\in \mathcal I_\mathrm{max}(F)$ and conversely $I\times J$ is clearly an independent set of $G\lor F$, that is,
\begin{equation}
 \mathcal I_\mathrm{max}(G\lor F)\subseteq \set{I\times J| I\in \mathcal I_\mathrm{max}(G),\; J\in \mathcal I_\mathrm{max}(F)}
 \subset \mathcal I(G\lor F).
\end{equation}
Additionally we use that the weight $\boldsymbol{w}$ is of product form, that is $w_{i,j} = w_i^G w_j^F$. Therefore, it follows that $\chi_f(G \lor F, \boldsymbol{w})$ is the solution to the problem
\begin{align}
\label{eq_chiGorFwrewritten2}
\begin{split}
\min  \quad & \sum_{I,J} z_{I,J} \\
\text{such that}\quad &\sum_{ I\ni u,J\ni v} z_{I,J} \geq w_u^G w_v^F \quad \forall u \in V(G), \ \forall v \in V(F)  \\
&  z_{I,J}  \geq 0.
\end{split}
\end{align}
By comparing this problem with the problems that define $\chi_f (G,\boldsymbol{w}^G)$ and $\chi_f (F,\boldsymbol{w}^F)$, we find that the product of the solutions that lead to $\chi_f(G,\boldsymbol{w}^G)$  and $\chi_f(F,\boldsymbol{w}^F)$ fulfil the conditions from Eq.~\eqref{eq_chiGorFwrewritten2}, since the solution of $\chi_f(G,\boldsymbol{w}^G)$ fulfils that $\sum_{ I\ni u} x_I \geq w_u \ \forall u \in V(G)$ and the solution of $\chi_f(F,\boldsymbol{w}^F)$ fulfils that $\sum_{ J\ni v} x_J \geq w_v \ \forall v \in V(F)$.
Therefore, we can conclude that
\begin{align}
    \label{eq_frac_chrom_number_GorF_leq_frac_G_or_frac_F}
    \chi_f(G, \boldsymbol{w}^G)   \chi_f(F, \boldsymbol{w}^F) \geq \chi_f(G \lor F,\boldsymbol{w}).
\end{align}

We obtain the other direction by considering the dual problem of the fractional chromatic number. Due to strong duality between dual linear programs, we know that $\chi_f(G \lor F, \boldsymbol{w}^{G\lor F})$ is also the solution to the dual problem. For $G\lor F$ the dual of the weighted fractional chromatic number is given by
\begin{align}
\label{eq_dual_weighted_fractional_chromatic_number_GorF}
\begin{split}
\max\quad & \sum_{i \in V(G\lor F )} w_i  z_i \\
\text{such that} \quad &\sum_{ i \in I} z_i \leq 1 \quad \forall I \in \mathcal{I}(G\lor F )\\
&  z_i  \geq 0 \quad  \forall i \in V(G\lor F )
\end{split}
\end{align}
and can be rewritten as
\begin{align}
\label{eq_dual_weighted_fractional_chromatic_number_GorF2}
\begin{split}
\max\quad & \sum_{u \in V(G),v \in V(F) } w^G_u  w^F_v  z_{u,v} \\
\text{such that} \quad &\sum_{ u \in I,v \in J} z_{u,v} \leq 1 \quad \forall I \in \mathcal{I}_\mathrm{max}(G), \ \forall J \in \mathcal{I}_\mathrm{max}(F) \\
&  z_{u,v}  \geq 0 \quad  \forall  u \in V(G), \ \forall v \in V(F),
\end{split}
\end{align}
since similar to before, the optimal solution can always be achieved by maximal independent sets. Following the same argument as before, namely that by considering the dual problem for $G$ and $F$ respectively, we can find a solution that fulfils the constraints from Eq.~\eqref{eq_dual_weighted_fractional_chromatic_number_GorF2}.  It follows that
\begin{align}
    \chi_f(G, \boldsymbol{w}^G)   \chi_f(F, \boldsymbol{w}^F) \leq \chi_f(G \lor F,\boldsymbol{w}).
\end{align}
Together with Eq.~\eqref{eq_frac_chrom_number_GorF_leq_frac_G_or_frac_F}, we conclude that
\begin{align}
      \chi_f(G \lor F,\boldsymbol{w}) = \chi_f(G, \boldsymbol{w}^G)   \chi_f(F, \boldsymbol{w}^F).
\end{align}

\section{Proof of Lemma~\ref{lem_disj.prod._of_SIC_and_odd_cyc_has_no_rep_in_dim_2_times_d}}
\label{app_odd_cycle_properties}

\begin{figure}
\centering
\includegraphics[width=0.4\linewidth]{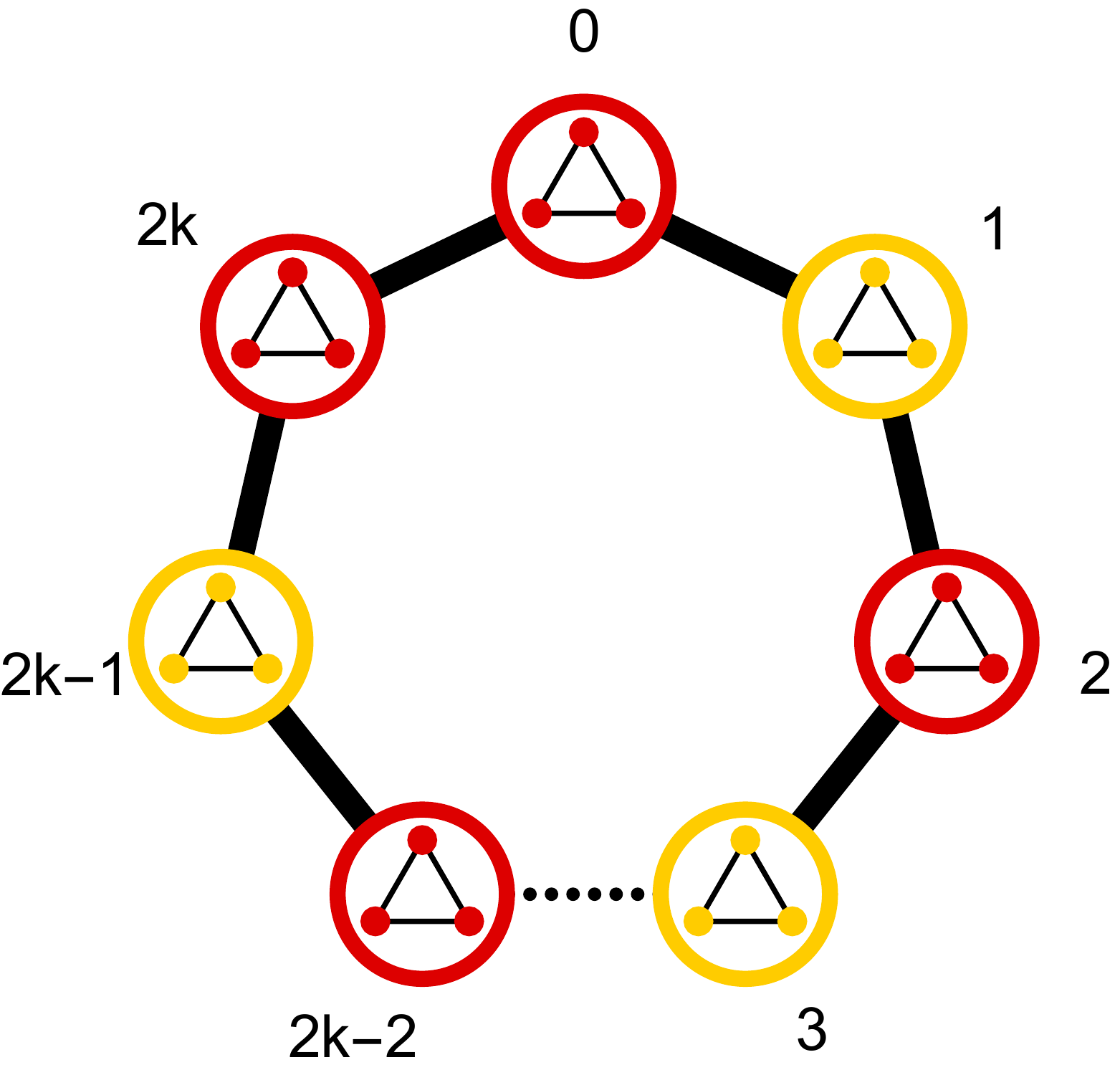}
\caption{Sketch of a subgraph of $G \lor C_{2k+1}$, where each triangle corresponds to a clique of size $d$. Every vertex $i$ is orthogonal to vertex $i\oplus 1$ and $i\oplus (-1)$, where ``$\oplus$'' denotes addition modulo $(2k+1)$. The two colours red and gold of the vertices correspond to two rank-$d$ projectors $P_0$ and $P_1$ that map into a subspace of dimension $d$, respectively and that must be orthogonal to each other. A contradiction occurs since there exists no colouring, such that the colour of every vertex is different to the colour of its neighbour. Therefore, $d(G\lor C_{2k+1}) \geq 2d+1$.}
\label{fig_Sketch_of_odd_cycle_with_SIC_graph}
\end{figure}

Let $S$ be any induced subgraph of $G$. Then $S \lor C_\ell$ is an induced subgraph of $G \lor C_\ell$. In addition, $d_{\pi}(S,r) \le d_{\pi}(G,r)$ since the smallest rank-$r$ PR of $G$ must already obey all orthogonality conditions of $S$. We choose now as induced subgraph $S$ a clique of $G$ with $d$ vertices. Then $G\lor C_\ell$ contains $(C_\ell)^d$ as subgraph due to
\begin{equation}\begin{split}
 E(G\lor C_\ell)&\supset E(S\lor C_\ell)
 =\set{ ((u,v),(\tilde u,\tilde v)) | u\ne \tilde u \text{ or } (v,\tilde v)\in E(C_\ell) }\\
 &\supset\set{((u,v),(\tilde u,\tilde v)) | (u\ne \tilde u \text{ and } v=\tilde v) \text{ or } (v,\tilde v)\in E(C_\ell) } =E((C_\ell)^d),
\end{split}\end{equation}
where the graph power on the right-hand side is defined in Ref.~\onlinecite{Schrivjer_Combinatorial_Optimization_p1096/97}, see also Eq.~(9) in Ref.~\onlinecite{State-ind-quantum-contextuality-with-projectors-of-nonunit-rank}. The resulting graph for $d=3$ is sketched in Figure~\ref{fig_Sketch_of_odd_cycle_with_SIC_graph}.

According to Theorem~1 in Ref.~\onlinecite{State-ind-quantum-contextuality-with-projectors-of-nonunit-rank} and using $((C_\ell)^d)^r=(C_\ell)^{dr}$, we have $d_{\pi}((C_\ell)^d,r) = d_{\pi}(C_\ell,dr)$.
Thus it suffices to prove that $C_\ell$ has no rank-$rd$ PR in a $2rd$-dimensional space and then choose $d=\omega(G)$. We proceed by contradiction and assume that $(P_i)_i$ is a rank-$rd$ PR of $C_\ell$. If this representation is in a $2dr$-dimensional space, we have $P_{i} + P_{i+1} = \openone$. Thus, we have $P_{i+2} = P_i$, which leads to $P_{\ell-1} = P_0$ since $\ell$ is odd. However, $P_{\ell-1}$ and $P_0$ should be orthogonal to each other, since $(0,\ell-1)$  is an edge in $C_\ell$.

\section{Graphs in {\it graph6} code}\label{appendix-graph6}
Here we provide a list of the graphs used in the manuscript in \emph{graph6} code. This code can be used in many computer software packages for graphs. Its description can be found in the software \emph{nauty} and at \url{http://cs.anu.edu.au/~bdm/data/formats.txt}.
\begin{tabbing}
$G_\mathrm{BBC}\quad$ \= \texttt{T??????wCcOcaOSGgWaWODS?IoHoH@BO\_eB?}\\
$G_\mathrm{YO}$ \> \texttt{L??G]OxhAgkOc\textasciigrave{}}\\
$G_\mathrm{H}$ \> \texttt{Qz[\textasciigrave{}MNFeCod\_K\_L?ODsAk\_KQ@H?}\\
$G_\mathrm{X}$ \> \texttt{QzrLDDBOxWD\textasciigrave{}TGUCEH@cG@T?Hc?}\\
$G_\mathrm{BBCr}$ \> \texttt{P?ACC@CCXAGPC\_H?dUAQEFB?}\\
$G_\mathrm{CEG}$ \> \texttt{QtaLc[gDBGkcUHUEG\textbackslash{}IElK\textbackslash{}OMy?}\\
\end{tabbing}

\bibliographystyle{quantum}
\bibliography{refs}

\end{document}